\documentclass{llncs}
\usepackage{graphicx}
\usepackage{amsmath}
\usepackage{amssymb}
\usepackage[tight,TABTOPCAP]{subfigure}
\usepackage{times}
\usepackage{picinpar}

\newcommand{\Tool}[1]{\textsf{\small #1}}
\newcommand{\STool}[1]{\textsf{\scriptsize #1}}
\newcommand{\Minisat}{\Tool{minisat}}
\newcommand{\SMinisat}{\STool{minisat}}
\newcommand{\Cryptominisat}{\Tool{cryptominisat}}

\newcommand{\substitution}[2]{#1/#2}
\newcommand{\simplification}[3]{\ensuremath{#1 \left[ \substitution{#2}{#3}\right] } }

\newcommand{\genrule}{\ensuremath{\textsf{$\oplus$\small-Gen}}}
\newcommand{\unitruleP}{\ensuremath{\textsf{\small$\oplus$-Unit$^+$}}}
\newcommand{\unitruleN}{\ensuremath{\textsf{\small$\oplus$-Unit$^-$}}}
\newcommand{\eqvruleP}{\ensuremath{\textsf{$\oplus$-Eqv$^+$}}}
\newcommand{\eqvruleN}{\ensuremath{\textsf{$\oplus$-Eqv$^-$}}}

\newcommand{\True}{\top}
\newcommand{\False}{\bot}

\newcommand{\Set}[1]{\{{#1}\}}
\newcommand{\Setdef}[2]{\ensuremath{\left\{{#1}\mid{#2}\right\}}}

\newcommand{\Tuple}[1]{{\ensuremath\langle{#1}\rangle}}
\newcommand{\Pair}[2]{{\ensuremath\langle{#1,#2}\rangle}}

\newcommand{\cnf}[1]{\operatorname{cnf}(#1)}

\newcommand{\X}{\oplus}
\newcommand{\Iff}{\Leftrightarrow}
\newcommand{\Implies}{\Rightarrow}

\newcommand{\defrule}[2]{
\begin{tabular}{@{}c@{}}
        \ensuremath{#1} \\
    \hline
        \ensuremath{#2} 
    \end{tabular}}

\newcommand{\inferencerule}[3]{\defrule{#1\qquad#2}{#3}}

\newcommand{\UPsys}{\ensuremath{\textup{\textsf{\small{UP}}}}}
\newcommand{\UPderiv}{\ensuremath{\mathrel{\vdash_{\textup{\textsf{\scriptsize{UP}}}}}}}
\newcommand{\ECAIsys}{\ensuremath{\textup{\textsf{\small{Subst}}}}}

\newcommand{\ECsys}{\ensuremath{\textup{\textsf\small{EC}}}}

\newcommand{\trail}[1]{\ensuremath{\pi_{#1}}}

\newcommand{\vars}{\operatorname{vars}}
\newcommand{\VarsOf}[1]{\vars(#1)}
\newcommand{\lits}{\operatorname{lits}}

\newcommand{\AL}{l}

\newcommand{\IL}{\hat{l}}

\newcommand{\KW}[1]{\textsf{\footnotesize{}{#1}}}
\newcommand{\While}{\KW{while}}
\newcommand{\For}{\KW{for}}

\newcommand{\If}{\KW{if}}
\newcommand{\Elif}{\KW{else if}}
\newcommand{\Else}{\KW{else}}
\newcommand{\Break}{\KW{break}}

\newcommand{\Comment}[1]{\textsf{\footnotesize/*{#1}*/}}
\newcommand{\Confl}{\mathit{confl}}

\newcommand{\Models}{\ensuremath{\models}}
\newcommand{\Conseq}{\ensuremath{\models}}
\newcommand{\set}[1]{\left\{{#1}\right\}}

\newcommand{\Form}{\ensuremath{\phi}}
\newcommand{\Cut}{W}

\newcommand{\RS}[1]{\operatorname{reasons}(#1)}
\newcommand{\Labfunc}{L}
\newcommand{\Lab}[1]{\Labfunc(#1)}

\newcommand{\clauses}[1]{\ensuremath{\phi_{#1}}}
\newcommand{\xorclauses}{\clauses{\textup{xor}}}

\newcommand{\orclauses}{\clauses{\textup{or}}}

\newcommand{\igraph}{G}
\newcommand{\vertices}{V}
\newcommand{\edges}{E}

\newcommand{\CExpl}[1]{\mathit{Expl}(#1)}
\newcommand{\CRLits}[2]{\operatorname{ReasonLits}_{\land}(#1,#2)}
\newcommand{\PExpl}[1]{\mathit{Expl_{\X}}(#1)}

\newcommand{\PRLits}[2]{\operatorname{ReasonLits}_{\X}(#1,#2)}

\newcommand{\NofPaths}[2]{\mathit{\#b\mbox{-}paths}(#1,#2)}

\newcommand{\tmpfname}[1]{f_{#1}}
\newcommand{\tmpf}[2]{\tmpfname{#2}(#1)}

\begin{document}

\title{Conflict-Driven XOR-Clause Learning \\(extended version)\thanks{This work has been financially supported by the Academy  of Finland under the Finnish Centre of Excellence in Computational Inference (COIN). The original version of the paper~\cite{LJN:SAT2012} has been presented in the 15th International Conference on Theory and Applications of Satisfiability Testing, SAT 2012.}
}

\author{Tero Laitinen, Tommi Junttila, and Ilkka Niemel\"a}
\institute{Aalto University\\
  Department of Information and Computer Science\\
  PO Box 15400, FI-00076 Aalto, Finland\\
  \email{\{Tero.Laitinen,Tommi.Junttila,Ilkka.Niemela\}@aalto.fi}
}

\maketitle

\begin{abstract}
  Modern conflict-driven clause learning (CDCL) SAT solvers
  are very good in solving conjunctive normal form (CNF) formulas.
  However, some application problems involve lots of parity (xor) constraints
  which are not necessarily efficiently handled if translated into CNF.
  This paper studies solving CNF formulas augmented with xor-clauses
  in the DPLL(XOR) framework
  where a CDCL SAT solver is coupled with a separate xor-reasoning module.
  New techniques for analyzing xor-reasoning derivations are developed,
  allowing one to obtain smaller CNF clausal explanations for xor-implied
  literals and also to derive and learn new xor-clauses.
  It is proven that these new techniques allow
  very short unsatisfiability proofs for some formulas whose CNF
  translations do not have polynomial size resolution proofs,
  even when a very simple xor-reasoning module
  capable only of unit propagation is applied.
  The efficiency of the proposed techniques is evaluated
  on a set of challenging logical cryptanalysis instances.
\end{abstract}

%--------------------------------------------------------------------------
%
%--------------------------------------------------------------------------
%
%--------------------------------------------------------------------------
\section{Introduction}

Modern propositional satisfiability (SAT) solvers (see e.g.~\cite{Handbook:CDCL}) have been
successfully applied in a number of industrial application domains.
Propositional satisfiability instances are typically encoded in conjunctive
normal form (CNF) which allows very efficient Boolean constraint propagation
and conflict-driven clause learning (CDCL) techniques.
However, such CNF encodings may not allow optimal exploitation of the problem
structure in the presence of parity (xor) constraints;
such constraints are abundant especially in the logical cryptanalysis domain
and also present in circuit verification and bounded model checking.
An instance consisting only of parity constraints can be solved in polynomial
time using Gaussian elimination,
but even state-of-the-art SAT solvers
relying only on basic Boolean constraint propagation and CDCL
can scale poorly on the corresponding CNF encoding.

In this paper we develop new techniques for exploiting structural properties
of xor constraints (i.e.~linear equations modulo 2)
in the recently introduced DPLL(XOR)
framework~\cite{LJN:ECAI2010,LJN:ICTAI2011}
where a problem instance is given as a combination of CNF and xor-clauses.
In the framework a CDCL SAT solver takes care of the CNF part while
a separate xor-reasoning module performs propagation on the xor-clauses.
In this paper we introduce new techniques for explaining
why a literal was implied or why a conflict occurred in the xor-clauses part;
such explanations are needed by the CDCL part.
The new core idea is to not see xor-level propagations as implications
but as linear arithmetic equations.
As a result, the new proposed \emph{parity explanation} techniques can
(i) provide smaller clausal explanations for the CDCL part,
and also
(ii) derive new xor-clauses that can then be learned in the xor-clauses part.
The goal of learning new xor-clauses is,
similarly to clause learning in CDCL solvers,
to enhance the deduction capabilities of the reasoning engine.
We introduce the new techniques on a very simple xor-reasoning module
allowing only unit propagation on xor-clauses
and
prove that, even when new xor-clauses are not learned,
the resulting system with parity explanations can efficiently
solve parity problems whose CNF translations are very hard for resolution.
We then show that the new parity explanation techniques also
extend to more general xor-reasoning modules,
for instance to the one in~\cite{LJN:ECAI2010}
capable of equivalence reasoning in addition to unit propagation.
Finally,
we experimentally evaluate the effect of the proposed techniques
on a challenging benchmark set modelling cryptographic attacks.
The proofs of Theorems \ref{Thm:PExpCorrectness}--\ref{Thm:PExpUrquhart}
are in the appendix.
%are provided % can be found
%at \url{http://users.ics.tkk.fi/tolaiti2/sat2012.pdf}.

\textbf{Related work.}
In~\cite{BaumgartnerMassacci:CL2000} a calculus combining basic DPLL without
clause learning and Gauss elimination is proposed;
their Gauss rules are similar to the general rule $\genrule$
we use in Sect.~\ref{Sec:General}.
The solvers EqSatz~\cite{Li:AAAI2000}, lsat~\cite{DBLP:conf/cp/OstrowskiGMS02}, and March\_eq~\cite{HeuleEtAl:SAT2004} incorporate parity reasoning into
DPLL without clause learning, extracting parity constraint information from
CNF and during look-ahead, and exploiting it during the preprocessing phase and  search.
MoRsat~\cite{Chen:SAT2009} extracts parity constraints from a CNF formula,
uses them for simplification during preprocessing,
and
proposes a watched literal scheme for unit propagation on parity constraints.
Cryptominisat~\cite{SoosEtAl:SAT2009,Soos}, like our approach,
accepts a combination of CNF and xor-clauses as input.
It uses the computationally relatively expensive Gaussian elimination as
the xor-reasoning method and
by default only applies it at the first levels of the search;
we apply lighter weight xor-reasoning at all search levels.
Standard clause learning is supported in MoRsat and Cryptominisat;
our deduction system characterization of xor-reasoning allows us
to exploit the linear properties
of xor-clauses to obtain smaller CNF explanations of xor-implied literals
and xor-conflicts
as well as to derive and learn new xor-clauses.

\section{Preliminaries}

An \emph{atom} is either a propositional variable or
the special symbol $\top$ which denotes the constant ``true''.
A \emph{literal} is an atom $A$ or its negation $\neg A$;
we identify $\neg\top$ with $\bot$ and $\neg \neg A$ with $A$.
A traditional, non-exclusive \emph{or-clause} is a disjunction
$l_1 \vee \dots \vee l_n$ of literals.
An {\it xor-clause} is an expression of form $ l_1 \oplus \dots \oplus l_n$,
where $l_1,\dots,l_n$ are literals and the symbol $\oplus$ stands for the
exclusive logical or.
In the rest of the paper,
we implicitly assume that each xor-clause is in a \emph{normal form}
such that (i) each atom occurs at most once in it, and
(ii) all the literals in it are positive. % i.e.~the negation does not appear.
The unique (up to reordering of the atoms) normal form
for an xor-clause can be obtained by applying the following rewrite rules
in any order until saturation: %is reached:
(i) ${{\neg A} \oplus C} \leadsto {A \oplus \top \oplus C}$, and
(ii) ${A \oplus A \oplus C} \leadsto C$,
where $C$ is a possibly empty xor-clause and $A$ is an atom.
For instance,
the normal form of ${\neg x_1} \oplus x_2 \oplus x_3 \oplus x_3$ is 
$x_1 \oplus x_2 \oplus \top$,
while the normal form of $x_1 \oplus x_1$ is the empty xor-clause $()$.
We say that an xor-clause is \emph{unary} if it is either of form $x$ or
$x \oplus \top$ for some variable $x$;
we will identify $x \oplus \top$ with the literal $\neg x$.
An xor-clause is \emph{binary} (\emph{ternary}) if
its normal form has two (three) variables.
A {\it clause} is either an or-clause or an xor-clause. 

A {\it truth assignment} $\pi$ is a set of literals such that
$\top \in \pi$ and $\forall l \in \pi : {\neg l} \notin \pi$.
We define the ``satisfies'' relation $\models$ between a truth assignment
$\pi$ and logical constructs as follows:
(i) if $l$ is a literal,
then $\pi \models l$ iff $l \in \pi$,
(ii) if $C = (l_1 \lor \dots \lor l_n)$ is an or-clause,
then $\pi \models C$ iff $\pi \models l_i$ for some $l_i \in \set{l_1,\ldots,l_n}$,
and
(iii)
if $C = (l_1 \oplus \dots \oplus l_n)$ is an xor-clause,
then $\pi \models C$ iff
$\pi$ is total for $C$ (i.e.~$\forall 1 \le i \le n : {l_i \in \pi} \lor {\neg l_i \in \pi}$)
and
$\pi \models l_i$ for an odd number of literals of $C$.
Observe that no truth assignment satisfies
the empty or-clause $()$ or the empty xor-clause $()$,
i.e.~these clauses are synonyms for $\bot$.

A \emph{cnf-xor formula} $\Form$ is a conjunction of clauses,
expressible as a conjunction
\begin{equation}
\Form = \orclauses \land \xorclauses,
\end{equation}
where $\orclauses$ is a conjunction of or-clauses and $\xorclauses$ is a conjunction of xor-clauses.
A truth assignment $\pi$ \emph{satisfies} $\Form$,
denoted by $\pi \models \Form$,
if it satisfies each clause in it;
$\Form$ is called \emph{satisfiable} if there exists such a truth assignment
satisfying it, and \emph{unsatisfiable} otherwise.
The \emph{cnf-xor satisfiability} problem studied in this paper is to decide
whether a given cnf-xor formula has a satisfying truth assignment.
A formula $\Form'$ is a \emph{logical consequence} of a formula $\Form$,
denoted by $\Form \Conseq \Form'$,
if $\pi \models \Form$ implies $\pi \models \Form'$
for all truth assignments $\pi$.
The set of variables occurring in a formula $\Form$ is denoted by
$\vars(\Form)$,
and 
$\lits(\Form) = \Setdef{x, \neg x}{x \in \vars(\Form)}$ is
the set of literals over $\vars(\Form)$.
We use $\simplification{C}{A}{D}$ to denote the (normal form) xor-clause that is
identical to $C$ except that all occurrences of the atom $A$ in $C$
are substituted with $D$ once.
For instance,
$\simplification{(x_1 \oplus x_2 \oplus x_3)}{x_1}{(x_1 \oplus x_3)} =
{x_1 \oplus x_3 \oplus x_2 \oplus x_3} =
{x_1 \oplus x_2}$.

\section{The DPLL(XOR) framework}

\newcommand{\VA}{V_{\textup{a}}}
\newcommand{\VB}{V_{\textup{b}}}

The idea in the DPLL(XOR) framework~\cite{LJN:ECAI2010}
for satisfiability solving of cnf-xor formulas
$\Form = \orclauses \land \xorclauses$
is similar to that in the DPLL($T$) framework for solving
satisfiability of quantifier-free first-order formulas modulo
a background theory $T$
(SMT, see e.g.\ \cite{NieuwenhuisEtAl:JACM06,Handbook:SMT}).
In DPLL(XOR),
see Fig.~\ref{fig:DPLLXOR} for a high-level pseudo-code,
one employs a conflict-driven clause learning (CDCL) SAT solver
(see e.g.~\cite{Handbook:CDCL})
to search for a satisfying truth assignment $\pi$
over all the variables in $\Form = \orclauses \land \xorclauses$.\footnote{See \cite{LJN:ECAI2010} for a discussion on handling ``xor-internal'' variables occurring in $\xorclauses$ but not in $\orclauses$.}
The CDCL-part takes care of the usual unit clause propagation on the cnf-part
$\orclauses$ of the formula (line 4 in Fig.~\ref{fig:DPLLXOR}),
conflict analysis and non-chronological backtracking (line 15--17),
and
heuristic selection of decision literals (lines 19--20) which
extend the current partial truth assignment $\pi$ towards a total one.

\begin{figure}[bt]
{\small
\begin{tabbing}
99.\={mm}\={mm}\={mm}\=\kill
solve($\Form = {\orclauses \land \xorclauses}$):\\
1.\>initialize xor-reasoning module $M$ with $\xorclauses$\\
2.\>$\pi = \Tuple{}$\qquad\Comment{the truth assignment}\\
3.\>\While{} true:\\
4.\>\>$(\pi',\Confl) = \textsc{unitprop}(\orclauses,\pi)$\quad\Comment{unit propagation}\\
5.\>\>\If{} not $\Confl$:\qquad\Comment{apply xor-reasoning}\\
6.\>\>\>\For{} each literal $l$ in $\pi'$ but not in $\pi$: $M$.\textsc{assign}($l$)\\
7.\>\>\>$(\IL_1,...,\IL_k) = M.\textsc{deduce}()$\\
8.\>\>\>\For{} $i=1$ to $k$:\\
9.\>\>\>\>$C = M.\textsc{explain}(\IL_i)$\\
10.\>\>\>\>\If{} $\IL_i = \False$ or ${\neg \IL_i} \in \pi'$: $\Confl = C$, \Break{}\\
11.\>\>\>\>\Elif{} $\IL_i \notin \pi'$: add $\IL_i$ to $\pi'$ with the implying or-clause $C$\\
12.\>\>\>\If{} $k > 0$ and not $\Confl$:\\
13.\>\>\>\>$\pi = \pi'$; continue\quad\Comment{unit propagate further}\\
14.\>\>let $\pi = \pi'$\\
15.\>\>\If{} $\Confl$:\qquad\Comment{standard Boolean conflict analysis}\\
16.\>\>\>analyze conflict, learn a conflict clause\\
17.\>\>\>backjump or return ``unsatisfiable'' if not possible\\
18.\>\>\Else{}: \\
19.\>\>\>add a heuristically selected unassigned literal in $\Form$ to $\pi$\\
20.\>\>\>or return ``satisfiable'' if no such variable exists
\end{tabbing}%
}%
\caption{The essential skeleton of the DPLL(XOR) framework}%
\label{fig:DPLLXOR}%
\end{figure}%

To handle the parity constraints in the xor-part $\xorclauses$,
an \emph{xor-reasoning module} $M$ is coupled with the CDCL solver.
The values assigned in $\pi$ to the variables in $\VarsOf{\xorclauses}$
by the CDCL solver are communicated as \emph{xor-assumption literals}
to the module (with the \textsc{assign} method on line 6 of the pseudo-code).
If $l_1,...,l_m$ are the xor-assumptions communicated to the module so far,
then the \textsc{deduce} method (invoked on line 7) of the module
is used to deduce a (possibly empty) list of \emph{xor-implied literals}
$\IL$ that are logical consequences of the xor-part $\xorclauses$ and
xor-assumptions,
i.e.~literals for which ${\xorclauses \land l_1 \land ... \land l_m} \Models \IL$ holds.
These xor-implied literals can then be added to the current truth assignment
$\pi$ (line 11) and the CDCL part invoked again to perform unit clause
propagation on these.
The conflict analysis engine of CDCL solvers requires that
each implied (i.e.~non-decision) literal has an \emph{implying clause},
i.e.~an or-clause that forces the value of the literal by unit propagation
on the values of literals appearing earlier in
the truth assignment (which at the implementation level is a sequence of
literals instead of a set).
For this purpose
the xor-reasoning module has a method \textsc{explain} that,
for each xor-implied literal $\IL$,
gives an or-clause $C$ of form ${l'_1 \land ... \land l'_k} \Implies \IL$,
i.e.~${\neg l'_1} \lor ... \lor {\neg l'_k} \lor \IL$,
such that
(i) $C$ is a logical consequence of $\xorclauses$,
and
(ii) $l'_1,...,l'_k$ are xor-assumptions made or xor-implied literals returned
before $\IL$.
An important special case occurs when the ``false'' literal $\False$ 
is returned as an xor-implied literal (line 10), i.e.~when an \emph{xor-conflict} occurs;
this implies that ${\xorclauses \land l_1 \land ... \land l_m}$
is unsatisfiable.
In such a case, the clause returned by the \textsc{explain} method
is used as the unsatisfied clause $\Confl$ initiating the conflict analysis
engine of the CDCL part (lines 10 and 15--17).

In addition to the \textsc{assign}, \textsc{deduce}, and \textsc{explain}
methods,
an xor-reasoning module must also implement methods that allow
xor-assumptions to be retracted from the solver in order to
allow backtracking in synchronization with the CDCL part (line 17).

Naturally,
there are many \emph{xor-module integration strategies}
that can be considered in addition to the one described in
the above pseudo-code.
For instance,
the xor-explanations for the xor-implied literals can be computed
always (as in the pseudo-code for the sake of simplicity)
or
only when needed in the CDCL-part conflict analysis
(as in a real implementation for efficiency reasons).

\section{The xor-reasoning module ``\UPsys''}

To illustrate our new parity-based techniques,
we first introduce a very simple xor-reasoning module ``\UPsys{}''
which only performs unit propagation on xor-clauses.
As such it can only perform the same deduction as CNF-level unit propagation
would on the CNF translation of the xor-clauses.
However, with our new parity-based xor-implied literal explanation techniques
(Sect.~\ref{Sec:ParityExplanations}) we can deduce much stronger clauses (Sect.~\ref{Sec:Resolution})
and
also new xor-clauses that can be learned (Sect.~\ref{Sec:Learning}).
In Sect.~\ref{Sec:General} we then generalize the results to other
xor-reasoning modules such as the the one in \cite{LJN:ECAI2010}
incorporating also equivalence reasoning.

\begin{figure}[t]
  \centering
  \small
  \begin{tabular}{l@{\qquad}l}
    \unitruleP:\ 
    $\inferencerule{x}{C}{\simplification{C}{x}{\top}}$
    &
    \unitruleN:\ 
    $\inferencerule{x \oplus \True}{C}{\simplification{C}{x}{\False}}$
    \\
  \end{tabular}%
  \caption{Inference rules of \UPsys{}; the symbol $x$ is variable and $C$ is an xor-clause}
  \label{Fig:UPRules}%
\end{figure}

As explained above,
given a conjunction of xor-clauses $\xorclauses$ and a
sequence $l_1,\dots,l_k$ of xor-assumption literals,
the goal of an xor-reasoning module is to deduce xor-implied literals and
xor-conflicts over $\psi = {\xorclauses \land l_1 \land \dots \land l_k}$.
To do this,
the \UPsys{}-module implements a deduction system with
the inference rules shown in Fig.~\ref{Fig:UPRules}.
%
%We define the deduction of the \UPsys{}-module formally as
An {\it \UPsys{}-derivation} on $\psi$ is a finite,
vertex-labeled directed acyclic graph $G = \Tuple{V,E,\Labfunc}$,
where each vertex $v\in V$ is labeled with an xor-clause $\Labfunc(v)$
and
the following holds for each vertex $v$:
\begin{enumerate}
\item
  $v$ has no incoming edges (i.e. is an {\it input vertex}) and
  $\Labfunc(v)$ is an xor-clause in $\psi$,
  or 
\item
  $v$ has two incoming edges originating from vertices $v_1$ and $v_2$,
  and $\Labfunc(v)$ is derived from $\Labfunc(v_1)$ and $\Labfunc(v_2)$
  by using one of the inference rules.
\end{enumerate}

\begin{figwindow}[2,r,%
\makebox[.49\textwidth][c]{\includegraphics[width=.47\textwidth]{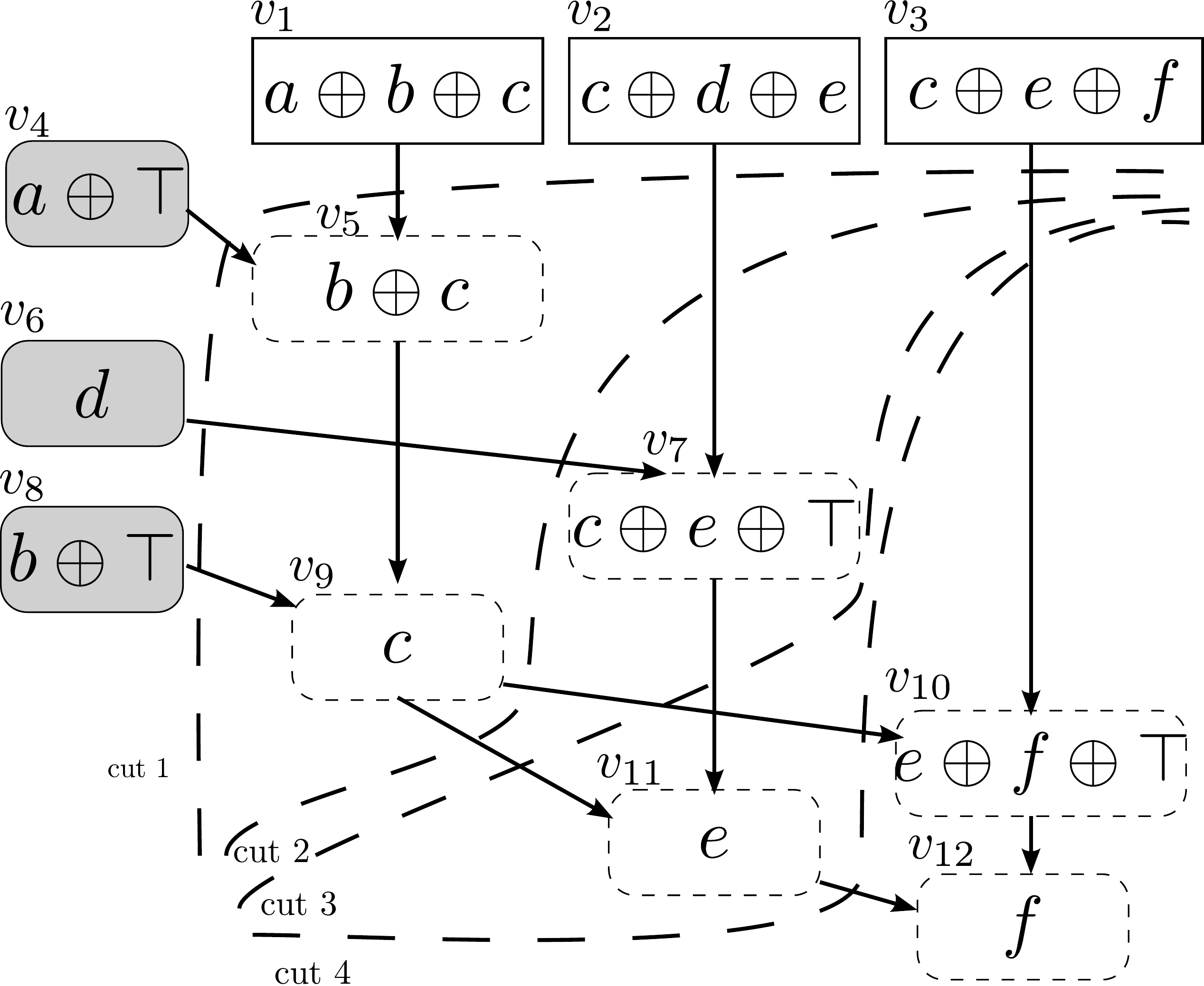}},%
{A $\UPsys$-derivation\label{Fig:UPDerivation}}]
As an example, Fig.~\ref{Fig:UPDerivation} shows a \UPsys-derivation
for $\xorclauses \land (\neg a) \land (d) \land (\neg b)$,
where $\xorclauses = (a \X b \X c) \land (c \X d \X e) \land (c \X e \X f)$
(please ignore the ``cut'' lines for now).
An xor-clause $C$ is \emph{$\UPsys$-derivable} on $\psi$,
denoted by $\psi \UPderiv C$,
if there exists a \UPsys{}-derivation on $\psi$ that contains a vertex
labeled with $C$;
the \UPsys-derivable unary xor-clauses are
the xor-implied literals that the \UPsys-module returns
when its \textsc{deduce} method is called.
In Fig.~\ref{Fig:UPDerivation},
the literal $f$ is $\UPsys$-derivable and
the $\UPsys$-module returns $f$ as an xor-implied literal after
$\neg a$, $d$, and $\neg b$ are given as xor-assumptions.
As a direct consequence of the definition of xor-derivations and
the soundness of the inference rules,
it holds that if an xor-derivation on $\psi$ contains a vertex labeled
with the xor-clause $C$, then $C$ is a logical consequence of $\psi$,
i.e.~$\psi \UPderiv C$ implies $\psi \Conseq C$.
A $\UPsys$-derivation on $\psi$ is a {\it $\UPsys$-refutation of $\psi$}
if it contains a vertex labeled with the false literal $\False$;
in this case, $\psi$ is unsatisfiable.
A $\UPsys$-derivation $G$ on $\psi$ is {\it saturated} if
for each unary xor-clause $C$ such that $\psi \UPderiv C$
it holds that there is a vertex $v$ in $G$ with the label $\Labfunc(v) = C $.
Note that $\UPsys$ is not refutationally complete,
e.g.~there is no $\UPsys$-refutation of the unsatisfiable conjunction
$(a \X b) \land (a \X b \X \True)$.
However, it is ``eventually refutationally complete'' in the DPLL(XOR) setting:
if each variable in $\psi$ occurs in a unary clause in $\psi$,
then the empty clause is $\UPsys$-derivable iff $\psi$ is unsatisfiable;
thus when the CDCL SAT solver has assigned a value to all variables in
$\xorclauses$,
the $\UPsys$-module can check whether all the xor-clauses are satisfied.
\end{figwindow}

As explained in the previous section,
the CDCL part of the DPLL(XOR) framework requires
an implying or-clause for each xor-implied literal.
These can be computed by interpreting the \unitruleP{} and \unitruleN{} rules
as implications
\begin{eqnarray}
(x) \land C &\Implies& \simplification{C}{x}{\True}\label{Eq:UIP}\\
(x \oplus \True) \land C &\Implies& \simplification{C}{x}{\False}\label{Eq:UIN}
\end{eqnarray}
respectively,
and
recursively expanding the xor-implied literal with the left-hand side
conjunctions of these until a certain cut of the \UPsys-derivation is reached.
Formally,
a \emph{cut} of a \UPsys-derivation
$\igraph = \Tuple{ \vertices, \edges, \Labfunc}$ is
a partitioning $(\VA,\VB)$ of $\vertices$.
A \emph{cut for a non-input vertex $v \in \vertices$}
is a cut $(\VA,\VB)$ such that
(i) $v \in \VB$, and
(ii) if $v' \in \vertices$ is an input vertex and
there is a path from $v'$ to $v$, then $v' \in \VA$.
Now assume a \UPsys-derivation
$\igraph = \Tuple{\vertices, \edges, \Labfunc}$
for $\xorclauses \land {\AL_1 \land ... \land \AL_k}$.
For each non-input node $v$ in $\igraph$,
and
each cut $\Cut=\Pair{\VA}{\VB}$ of $\igraph$ for $v$,
the \emph{implicative explanation} of $v$ under $\Cut$
is the conjunction $\CExpl{v,\Cut} = \tmpf{v}{\Cut}$,
there $\tmpfname{\Cut}$ is recursively defined as follows:
\begin{itemize}
\item[E1]
  If $u$ is an input node with $\Lab{u} \in \xorclauses$,
  then $\tmpf{u}{\Cut} = \top$.
\item[E2]
  If $u$ is an input node with $\Lab{u} \in \Set{\AL_1,...,\AL_k}$,
  then $\tmpf{u}{\Cut} = \Lab{u}$.
\item[E3]
  If $u$ is a non-input node in $\VA$,
  then $\tmpf{u}{\Cut} = \Lab{u}$.
\item[E4]
  If $u$ is a non-input node in $\VB$,
  then
  $\tmpf{u}{\Cut} =
  {\tmpf{u_1}{\Cut} \land \tmpf{u_2}{\Cut}}$,
  where $u_1$ and $u_2$ are the source nodes of the two edges incoming to $u$.
\end{itemize}
Based on Eqs.~\eqref{Eq:UIP} and \eqref{Eq:UIN}, it is easy to see that
$\xorclauses \Conseq \CExpl{v,\Cut} \Implies \Lab{v}$ holds.
The implicative explanation $\CExpl{v,\Cut}$ can in fact be read
directly from the cut $\Cut$ %for $v$
as in~\cite{LJN:ECAI2010}:
$\CExpl{v,\Cut} = \bigwedge_{u \in \RS{\Cut}} \Lab{u}$,
where
$\RS{\Cut} = \{{u \in \VA} \mid {\Lab{u} \notin \xorclauses} \land {\exists u' \in \VB : \Tuple{u,u'} \in \edges}\}$
is the \emph{reason set} for $\Cut$.
%
%The \emph{reason set} for $\Cut$ is defined as
%$\RS{\Cut} = \Setdef{u \in \VA}{\exists u' \in \VB : \Tuple{u,u'} \in \edges}$
%and
%the implicative explanation $\CExpl{v,\Cut}$ can in fact be read
%directly from the cut $\Cut$ for $v$ as in~\cite{LJN:ECAI2010}:
%$\CExpl{v,\Cut} = \bigwedge_{{u \in \RS{\Cut}} \land {\Lab{u} \notin \xorclauses}} \Lab{u}$.
%
%Define that the \emph{reason set} for a cut $\Cut=(\VA,\VB)$ is
%$\RS{\Cut} = \Setdef{u \in \VA}{\exists u' \in \VB : \Tuple{u,u'} \in \edges}$.
%
A cut $\Cut$ is \emph{cnf-compatible} if
$\Lab{u}$ is a unary xor-clause
for each $u \in \RS{\Cut}$.
%either
%(i) $u$ is an input vertex, or
%
%A cut $\Cut$ is \emph{cnf-compatible} if
%for each $u \in \RS{\Cut}$ either
%(i) $u$ is an input vertex, or
%(ii) $\Lab{u}$ is a unary clause.
%
Thus if the cut $\Cut$ is cnf-compatible,
then $\CExpl{v,\Cut} \Implies \Lab{v}$ is the required or-clause
implying the xor-implied literal $\Lab{v}$.

\begin{example}
%As an example,
Consider again the \UPsys{}-derivation on
$\xorclauses \land (\neg a) \land (d) \land (\neg b)$ in
Fig.~\ref{Fig:UPDerivation}.
It has four cuts, 1--4, for the vertex $v_{12}$,
corresponding to the explanations 
${\neg a} \land d \land {\neg b}$,
$c \land d$,
$c \land (c \oplus e \oplus \top)$,
and
$e \land c$, respectively.
The non-cnf-compatible cut 3 cannot be used to give an implying or-clause
for the xor-implied literal $f$ but the others can;
the one corresponding to the cut 2 is $({\neg c} \lor {\neg d} \lor f)$.
\hfill$\clubsuit$
\end{example}

The UP-derivation bears an important similarity with ``traditional''
implication graph of a SAT solver where each vertex represents a variable
assignment: graph partitions are used to derive clausal explanations
for implied literals.
Different partitioning schemes for such implication
graphs have been studied in~\cite{DBLP:conf/iccad/ZhangMMM01},
and we can directly adopt some of them for our analysis.
A cut $\Cut = (\VA,\VB)$ for a non-input vertex $v$ is:
\begin{enumerate}
\item
  {\it closest cut} if $W$ is the cnf-compatible cut with the smallest possible
  $\VB$ part.
  Observe that each implying or-clause derived from these cuts
  is a clausification of a single xor-clause;
  e.g.,~$({\neg c} \lor {\neg e} \lor f)$ obtained from the cut 4 in Fig.~\ref{Fig:UPDerivation}.
\item
  \emph{first UIP cut} if $W$ is the %cnf-compatible
  cut with the largest possible $\VA$ part
  such that $\RS{\Cut}$ contains either
  the latest xor-assumption vertex or exactly one of its successors.
\item
  {\it furthest cut} if $\VB$ is maximal.
  Note that furthest cuts are also cnf-compatible as their reason sets
  consist only of xor-assumptions.
%\item
%  {\it furthest cut} if $\VB$ is maximal.
%  Note that furthest cuts are also cnf-compatible as their reason sets
%  consist only of original xor-clauses and xor-assumptions.
\end{enumerate}

In the implementation of the \UPsys-module,
we use a modified version of the 2-watched literals scheme first presented
in~\cite{DBLP:conf/dac/MoskewiczMZZM01} for or-clauses;
all but one of the \emph{variables} in an xor-clause need to be assigned before
the xor-clause implies the last one.
Thus it suffices to have two \emph{watched variables}. 
MoRsat~\cite{Chen:SAT2009} uses the same data structure for all
clauses and has $2 \times 2$ watched literals for xor-clauses.
Cryptominisat~\cite{SoosEtAl:SAT2009} uses a scheme similar to ours except
that it manipulates the polarities of literals in an xor-clause
while we take the polarities into account in the explanation phase.
Because of this implementation technique,
the implementation does not consider the non-unary non-input vertices
in \UPsys-derivations;
despite this, Thm.~\ref{Thm:PExpUrquhart} does hold also for the implemented
inference system.

\section{Parity explanations}
\label{Sec:ParityExplanations}

So far in this paper,
as well as in our previous works~\cite{LJN:ECAI2010,LJN:ICTAI2011},
we have used the inference rules %in Fig.~\ref{Fig:UPRules}
in an ``implicative way''.
For instance,
we have implicitly read the $\unitruleP$ rule as
\begin{quote}
  if the xor-clauses $(x)$ \emph{and} $C$ hold,
  \emph{then} $\simplification{C}{x}{\top}$ also holds.
\end{quote}
Similarly, 
the implicative explanation 
for an xor-implied literal $\IL$ labelling a non-input node $v$
under a cnf-compatible cut $\Cut$
has been defined to be a conjunction $\CExpl{v,\Cut}$ of literals
with $\xorclauses \Models {\CExpl{v,\Cut} \Rightarrow \IL}$ holding.
We now propose an alternative method allowing us to
compute a \emph{parity explanation} $\PExpl{v,\Cut}$
that is an xor-clause such that
\[
\xorclauses \Models {\PExpl{v,\Cut} \Iff \IL}\]
holds.
The variables in $\PExpl{v,\Cut}$ will always be a subset of the variables
in the implicative explanation $\CExpl{v,\Cut}$ computed on the same cut.

The key observation for computing parity explanations
is that the inference rules can
in fact also be read as \emph{equations} over xor-clauses
under some provisos.
As an example,
the $\unitruleP$ rule can be seen as the equation
$(x) \oplus C \oplus \top
 \Iff
 {\simplification{C}{x}{\top}}$
\emph{provided that}
(i) $x \in C$, and
(ii) $C$ is in normal form.
That is,
taking the exclusive-or of the two premises and the constant true
gives us the consequence clause of the rule.
The provisos are easy to fulfill:
(i) we have already assumed all xor-clauses to be in normal form,
and
(ii) applying the rule when $x \notin C$ is redundant and
can thus be disallowed.
%thus we can restrict the rule to cases when $x \in C$.
%
The reasoning is analogous for the \unitruleN{} rule
and
thus for \UPsys{} rules we have the equations:
\begin{eqnarray}
  %\unitruleP:
  (x) \oplus C \oplus \top 
  & \Iff &
  \simplification{C}{x}{\top}
  \label{Eq:PEP}
  \\
  %\unitruleN:
  (x \oplus \top) \oplus C \oplus \top
  & \Iff &
  \simplification{C}{x}{\bot}
  \label{Eq:PEN}
\end{eqnarray}

As all the \UPsys-rules can be interpreted as equations
of form ``left-premise xor right-premise xor true equals consequence'',
we can expand any xor-clause $C$ in a node of a \UPsys-derivation
by iteratively replacing it with the left hand side of
the corresponding equation.
As a result, we will get an xor-clause that is logically equivalent to $C$;
from this, we can eliminate the xor-clauses in $\xorclauses$ and
get an xor-clause $D$ such that $\xorclauses \Models {D \Iff C}$.
Formally,
assume a \UPsys-derivation
$\igraph = \Tuple{\vertices, \edges, \Labfunc}$
for $\xorclauses \land \AL_1 \land ... \land \AL_k$.
For each non-input node $v$ in $\igraph$,
and
each cut $\Cut=\Pair{\VA}{\VB}$ of $\igraph$ for $v$,
the \emph{parity explanation} of $v$ under $\Cut$
is $\PExpl{v,\Cut} = \tmpf{v}{\Cut}$,
there $\tmpfname{\Cut}$ is recursively defined as earlier
for $\CExpl{v,\Cut}$ except that the case ``E4'' is replaced by
\begin{itemize}
\item[E4]
  If $u$ is a non-input node in $\VB$,
  then
  $\tmpf{u}{\Cut} =
  {\tmpf{u_1}{\Cut} \X \tmpf{u_2}{\Cut} \X \top}$,
  where $u_1$ and $u_2$ are the source nodes of the two edges incoming to $u$.
\end{itemize}

We now illustrate parity explanations and
show that they can be smaller (in the sense of containing fewer variables)
than implicative explanations:
\begin{example}
  \label{ex:parityexp}
  Consider again the \UPsys-derivation
  given in Fig.~\ref{Fig:UPDerivation}.
  Take the cut 4 first;
  we get $\PExpl{v_{12},\Cut} = c \X e \X \top$.
  Now
  $\xorclauses \Models {\PExpl{v_{12},\Cut} \Iff \Lab{v_{12}}}$
  holds as $(c \X e \X \top) \Iff f$, i.e.~$c \X e \X f$,
  is an xor-clause in $\xorclauses$.
  Observe that the implicative explanation $c \land e$
  of $v_{12}$ under the cut is just one conjunct in the
  disjunctive normal form $(c \land e) \lor (\neg c \land \neg e)$
  of $c \X e \X \top$.

  On the other hand,
  under the cut 2 we get $\PExpl{v_{12},\Cut} = d$.
  Now
  $\xorclauses \Models {\PExpl{v_{12},\Cut} \Iff \Lab{v_{12}}}$
  as $d \Iff f$, i.e.~$d \X f \X \top$,
  is a linear combination of the xor-clauses in $\xorclauses$.
  Note that the implicative explanation for $v_{12}$ under the cut %2
  is $(c \land d)$,
  and
  no cnf-compatible cut for $v_{12}$ gives
  the implicative explanation $(d)$ for $v_{12}$. 
  \hfill$\clubsuit$
\end{example}
We observe that $\VarsOf{\PExpl{v,\Cut}} \subseteq \VarsOf{\CExpl{v,\Cut}}$
by comparing the definitions of $\CExpl{v,\Cut}$ and $\PExpl{v,\Cut}$.
The correctness of $\PExpl{v,\Cut}$, formalized in the following theorem,
can be established by induction and using Eqs.~\eqref{Eq:PEP} and \eqref{Eq:PEN}.
\begin{theorem}\label{Thm:PExpCorrectness}
  Let $\igraph = \Tuple{ \vertices, \edges, \Labfunc}$ be a \UPsys-derivation
  on $\xorclauses \wedge \AL_1 \wedge \dots \wedge \AL_k$,
  $v$ a node in it,
  and
  $\Cut = \Pair{\VA}{\VB}$ a cut for $v$.
  It holds that $\xorclauses \Models {\PExpl{v,\Cut} \Iff \Lab{v}}$.
\end{theorem}

Recall that the CNF-part solver requires an implying or-clause $C$ for each
xor-implied literal, forcing the value of the literal by unit propagation.
A parity explanation
can be used to get such implying or-clause
by taking the implicative explanation as a basis and
omitting the literals on variables not occurring in the parity explanation:
\begin{theorem}\label{Thm:PExpDisjunction}
  Let $\igraph = \Tuple{ \vertices, \edges, \Labfunc}$ be a \UPsys-derivation
  on $\xorclauses \land {\AL_1 \land \dots \land \AL_k}$,
  $v$ a node with $\Lab{v}=\IL$ in it,
  and
  $\Cut = \Pair{\VA}{\VB}$ a cnf-compatible cut for $v$.
  Then
  $\xorclauses \Models {(\bigwedge_{u \in S} \Lab{u}) \Implies \IL}$,
  where $S = \Setdef{u \in \RS{\Cut}}{\VarsOf{\Lab{u}} \subseteq \VarsOf{\PExpl{v,\Cut}}}$.
  %$\xorclauses \Models {(\bigwedge_{{u \in \RS{\Cut}} \land {\Lab{u} \notin \xorclauses} \land {\VarsOf{\Lab{u}} \subseteq \VarsOf{\PExpl{v,\Cut}}}} \Lab{u}) \Implies \IL}$.
\end{theorem}

\newcommand{\bpath}{b-path}
\newcommand{\Succ}[1]{\mathit{succ}(#1)}
\newcommand{\NofSuccPaths}[2]{\mathit{\#succ\mbox{-}b\mbox{-}paths}(#1,#2)}

Observing that only expressions of the type $\tmpf{u}{\Cut}$ occurring an odd
number of times in the expression $\tmpf{v}{\Cut}$ remain in $\PExpl{v,\Cut}$,
we can derive a more efficient graph traversal method for
computing parity explanations.
That is, when computing a parity explanation for a node,
we traverse the derivation backwards from it in a breadth-first order.
If we come to a node $u$ and note that its traversal is requested because
an even number of its successors have been traversed,
then we don't need to traverse $u$ further or
include $\Lab{u}$ in the explanation if $u$ was on the ``reason side'' $\VA$ of the cut.
\begin{example} 
  Consider again the \UPsys{}-derivation in Fig.~\ref{Fig:UPDerivation} and the
  cnf-compatible cut 1 for $v_{12}$.
  When we traverse the derivation backwards,
  we observe that the node $v_9$ has an even number of traversed successors;
  we thus don't traverse it (and consequently neither $v_8$, $v_5$, $v_4$ or $v_1$).
  On the other hand, $v_6$ has an odd number of traversed successors and
  it is included when computing $\PExpl{v_{12},\Cut}$.
  Thus we get $\PExpl{v_{12},\Cut} = \Labfunc(v_6) = (d)$ and
  the implying or-clause for $f$ is $d \Implies f$, i.e.~$({\neg d} \lor f)$.
  \hfill$\clubsuit$
\end{example}

Although parity explanations can be computed quite fast using graph traversal
as explained above, this can still be computationally prohibitive on
``xor-intensive'' instances because a single CNF-level conflict analysis may
require that implying or-clauses for hundreds of xor-implied literals are
computed.
In our current implementation, we compute the closest cnf-compatible cut (for
which parity explanations are very fast to compute but equal to
implicative explanations and produce clausifications of single
xor-clauses as implying or-clauses) for an xor-implied literal $\IL$
when an explanation is needed in the regular conflict analysis.
The computationally more expensive furthest cut is used if an explanation is
needed again in the conflict-clause minimization phase of minisat.

%--------------------------------------------------------------------------
%
%--------------------------------------------------------------------------
%
%--------------------------------------------------------------------------
\section{Resolution cannot polynomially simulate parity explanations}
\label{Sec:Resolution}

\newcommand{\charge}{c}
\newcommand{\var}{\operatorname{var}}
\newcommand{\urquhartgraph}{parity graph}
\newcommand{\urquhartclauses}{\operatorname{clauses}}
\newcommand{\urquhartxorclause}{\alpha}
\newcommand{\urquhartxorclauses}{\operatorname{xorclauses}}

Intuitively,
as parity explanations can contain fewer variables than implicative explanations,
the implying or-clauses derived from them should help pruning the remaining search space of the CDCL solver better.
We now show that,
in theory,
parity explanations can indeed be very effective as they can allow
small refutations for some formula classes whose
CNF translations do not have polynomial size resolution proofs.
To do this,
we use the hard formulas defined in~\cite{DBLP:journals/jacm/Urquhart87};
these are derived from a class of graphs which
we will refer to as ``parity graphs''.
A {\it \urquhartgraph{}} is an undirected, connected, edge-labeled graph
$ G = \Tuple{V,E} $
where each node $v \in V$ is labeled with a {\it charge} $\charge(v) \in
\set{\bot,\top}$ and each edge $\Tuple{v,u} \in E$ is labeled with a distinct
variable.
The {\it total charge} $\charge(G) = \bigoplus_{v \in V} \charge(v)$
of an \urquhartgraph{} $G$ is the parity of all node charges.
Given a node $v$,
define the xor-clause $\urquhartxorclause(v) = {q_1 \X \dots \X q_n \X \charge(v) \X \True}$,
where $q_1, \dots, q_n$ are the variables used as labels in the edges connected
to $v$,
and $\urquhartxorclauses(G) = \bigwedge_{v \in V} \urquhartxorclause(v)$.
For an xor-clause $C$ over $n$ variables,
let $\cnf{C}$ denote the equivalent CNF formula,
i.e.~the conjunction of $2^{n-1}$ clauses with $n$ literals in each.
Define $\urquhartclauses(G) = \bigwedge_{v \in V} \cnf{\urquhartxorclause(v)}$.

As proven in Lemma 4.1 in~\cite{DBLP:journals/jacm/Urquhart87},
$\urquhartxorclauses(G)$ and $\urquhartclauses(G)$
are unsatisfiable if and only if $\charge(G) = \True$.
The unsatisfiable formulas derived from parity graphs can be very hard
for resolution:
there is an infinite sequence $G_1,G_2,\dots$ of degree-bounded
\urquhartgraph{}s such that $\charge(G_i) = \True$ for each $i$
and
the following holds:
\begin{lemma}[Thm.~5.7 of~\cite{DBLP:journals/jacm/Urquhart87}]
  There is a constant $c > 1$ such that for sufficiently large $m$,
  any resolution refutation of $\urquhartclauses(G_m)$ contains
  $c^n$ distinct clauses, where $\urquhartclauses(G_m)$ is of length ${\cal O}(n)$, $n = m^2$.
\end{lemma}

We now present our key result on parity explanations:
for \emph{any} parity graph $G$ with $\charge(G_i) = \True$,
the formula $\urquhartxorclauses(G)$ can be refuted with a \emph{single}
parity explanation after a number of xor-assumptions have been made:
\begin{theorem}\label{Thm:PExpUrquhart}
  Let $G = \Tuple{V,E}$ be a %3-regular
  \urquhartgraph{} such that $c(G) = \True$.
  There is a \UPsys-refutation for $\urquhartxorclauses(G) \wedge q_1 \dots \wedge q_k $ for some xor-assumptions $ q_1,\dots,q_k$,
  a node $v$ with $\Lab{v}=\bot$ in it,
  and a cut $\Cut = \Tuple{\VA,\VB} $ for $v$
  such that $\PExpl{v,\Cut} = \top$.
  Thus $\urquhartxorclauses(G) \Models (\top \Iff \bot)$,
  showing $\urquhartxorclauses(G)$ unsatisfiable.
\end{theorem}

By recalling that CDCL SAT solvers are equally powerful to
resolution~\cite{PipatsrisawatDarwiche:AI2011},
and
that unit propagation on xor-clauses can be efficiently simulated by unit propagation their CNF translation,
we get the following:
\begin{corollary}
  There are families of unsatisfiable cnf-xor formulas for which
  DPLL(XOR) using $\UPsys{}$-module
  (i) has polynomial sized proofs if parity explanations are allowed,
  but
  (ii) does not have such %polynomial size proofs
  if the ``classic'' implicative explanations are used.
\end{corollary}
In practice, the CDCL part does not usually make the correct xor-assumptions
needed to compute the empty implying or-clause,
but if parity explanations are used in learning
as explained in the next section,
instances generated from parity graphs can be solved very fast.

%--------------------------------------------------------------------------
%
%--------------------------------------------------------------------------
%
%--------------------------------------------------------------------------
\section{Learning parity explanations}
\label{Sec:Learning}

As explained in Sect.~\ref{Sec:ParityExplanations},
parity explanations can be used to derive implying or-clauses,
required by the conflict analysis engine of the CDCL solver,
that are shorter than those derived by the classic implicative explanations.
In addition to this,
parity explanations can be used to derive \emph{new xor-clauses} that
are logical consequences of $\xorclauses$;
these xor-clauses $D$ can then be \emph{learned},
meaning that $\xorclauses$ is extended to $\xorclauses \land D$,
the goal being to increase the deduction power of the xor-reasoning module.
As an example,
consider again Ex.~\ref{ex:parityexp} and
recall that the parity explanation for $v_{12}$ under the cut 2 is $d$.
Now $\xorclauses \Models (d \Iff f)$,
i.e.~$\xorclauses \Models (d \X f \X \True)$,
holds,
and we can extend $\xorclauses$ to $\xorclauses' = \xorclauses \land (d \X f \X \True)$
while preserving all the satisfying truth assignments.
In fact,
it is not possible to deduce $f$ from $\xorclauses \land (d)$ by using \UPsys,
but $f$ can be deduced from $\xorclauses' \wedge (d)$.
Thus learning new xor-clauses derived from parity explanations
can increase the deduction power of the \UPsys{} inference system
in a way similar to conflict-driven clause learning increasing
the power of unit propagation in CDCL SAT solvers.

However,
if all such derived xor-clauses are learned,
it is possible to learn the same xor-clause many times,
as illustrated in the following example and
Fig.~\ref{Fig:Trail1}.
\begin{example}
  Let $\xorclauses = (a \X b \X c \X \True) \land (b \X c \X d \X e) \land ...$
  and
  assume that CNF part solver gives its first decision level literals
  $a$ and $\neg c$ as xor-assumptions to the \UPsys-module;
  the module deduces $b$ and returns it to the CNF solver.
  At the next decision level the CNF part guesses $d$, gives it to
  \UPsys-module, which deduces $e$, returns it to the CNF part, and
  the CNF part propagates it so that a conflict occurs.
  Now the xor-implied literal $e$ is explained and
  a new xor-clause $D = (a \X d \X e \X \top)$ is learned in $\xorclauses$.
  After this the CNF part backtracks,
  implies $\neg d$ at the decision level 1, and
  gives it to the \UPsys-module;
  the module can then deduce $\neg e$ \emph{without} using $D$.
  If $\neg e$ is now explained,
  the same ``new'' xor-clause $(a \X d \X e \X \top)$ can be derived.
  \hfill$\clubsuit$
\end{example}

\begin{figure}[t]
  \centering
  \includegraphics[width=\textwidth]{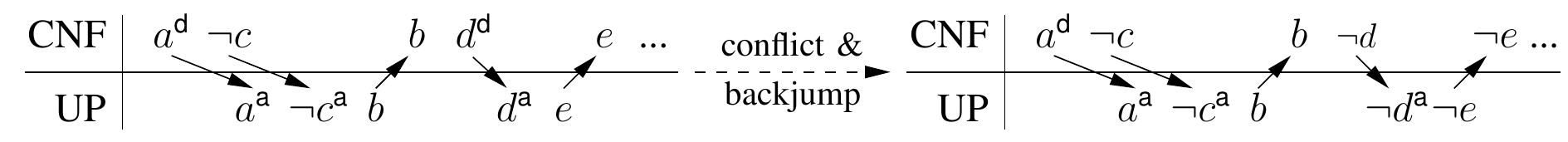}
  \caption{Communication between CNF part and \UPsys-module in a case when duplicate xor-clauses are learned; the \textsf{\scriptsize d} and \textsf{\scriptsize a} superscripts denote decision literals and xor-assumptions, respectively.}
\label{Fig:Trail1}
\end{figure}

The example illustrates a commonly occurring case in which a
derived xor-clause contains two or more literals on the latest decision level
($e$ and $d$ in the example);
in such a case,
the xor-clause may already exist in $\xorclauses$.
A conservative approach to avoid learning the same xor-clause twice,
under the reasonable assumption that the CNF and xor-reasoning module
parts saturate their propagations before new heuristic decisions are made,
is to disregard derived xor-clauses that have two or more
variables assigned on the latest decision level.
If a learned xor-clause for xor-implied literal $\IL$
does not have other literals on the latest decision level, it can be used
to infer $\IL$ with fewer decision literals.
Note that it may also happen that an implying or-clause for
an xor-implied literal $\IL$ does not contain any literals besides $\IL$
on the latest decision level;
the CNF part may then compute a conflict clause that does not have any
literals on the current decision level, which needs to be treated appropriately.

In order to avoid slowing down propagation in our implementation, we store
and remove learned xor-clauses using a strategy adopted from $\Minisat{}$: the
maximum number of learned xor-clauses is increased at each restart and the
``least active'' learned xor-clauses are removed when necessary. However, using
the conservative approach to learning xor-clauses, the total number of learned
xor-clauses rarely exceeds the number of original xor-clauses.

\section{General xor-derivations}
\label{Sec:General}

So far in this paper we have considered
a very simple xor-reasoning module capable only of unit propagation.
We can in fact extend the introduced concepts to more general inference
systems and derivations.
Define an \emph{xor-derivation} similarly to \UPsys{}-derivation except
that there is only one inference rule,
$\genrule: \inferencerule{C_1}{C_2}{C_1 \oplus C_2 \oplus \top}$,
where $C_1$ and $C_2$ are xor-clauses. The inference rule $\genrule$ is a generalization of the rules $\mbox{Gauss}^-$ and $\mbox{Gauss}^+$ in~\cite{BaumgartnerMassacci:CL2000}.
Now Thms.~\ref{Thm:PExpCorrectness} and \ref{Thm:PExpDisjunction}
can be shown to hold for such derivations as well.

As another concrete example of xor-reasoning module implementing a sub-class of \genrule,
consider the \ECAIsys{} module presented in~\cite{LJN:ECAI2010}.
In addition to the unit propagation rules of \UPsys{} in Fig.~\ref{Fig:UPRules},
it has inference rules allowing equivalence reasoning:
\[
\begin{array}{l@{\qquad\qquad}l}
  \eqvruleP:\ 
  $\inferencerule{x \oplus y \oplus \True}{C}
  {\simplification{C}{x}{y}}$
  &
  \eqvruleN:\ 
  $\inferencerule{x \oplus y}{C}
  {\simplification{C}{x}{(y\oplus\True)}}$
\end{array}
\]
where the symbols $x$ and $y$ are variables while
$C$ is an xor-clause in the normal form with an occurrence of $x$.
Note that these \ECAIsys{} rules are indeed instances of
the more general inference rule \genrule.
For instance,
given two xor-clauses
$C_1 = (c \oplus d \oplus \top) $ and $C_2 = (b \oplus d \oplus e)$,
the \ECAIsys{}-system can produce the xor-clause
$\simplification{C_2}{d}{c} = (b \oplus c \oplus e)$
which is also inferred by \genrule:
$ (C_1 \oplus C_2 \oplus \top) = ((c \oplus d \oplus \top) \oplus (b \oplus d \oplus e) \oplus \top) = (b \oplus c \oplus e) $.

\begin{figwindow}[2,r,%
\makebox[.47\textwidth][c]{\includegraphics[width=.43\textwidth]{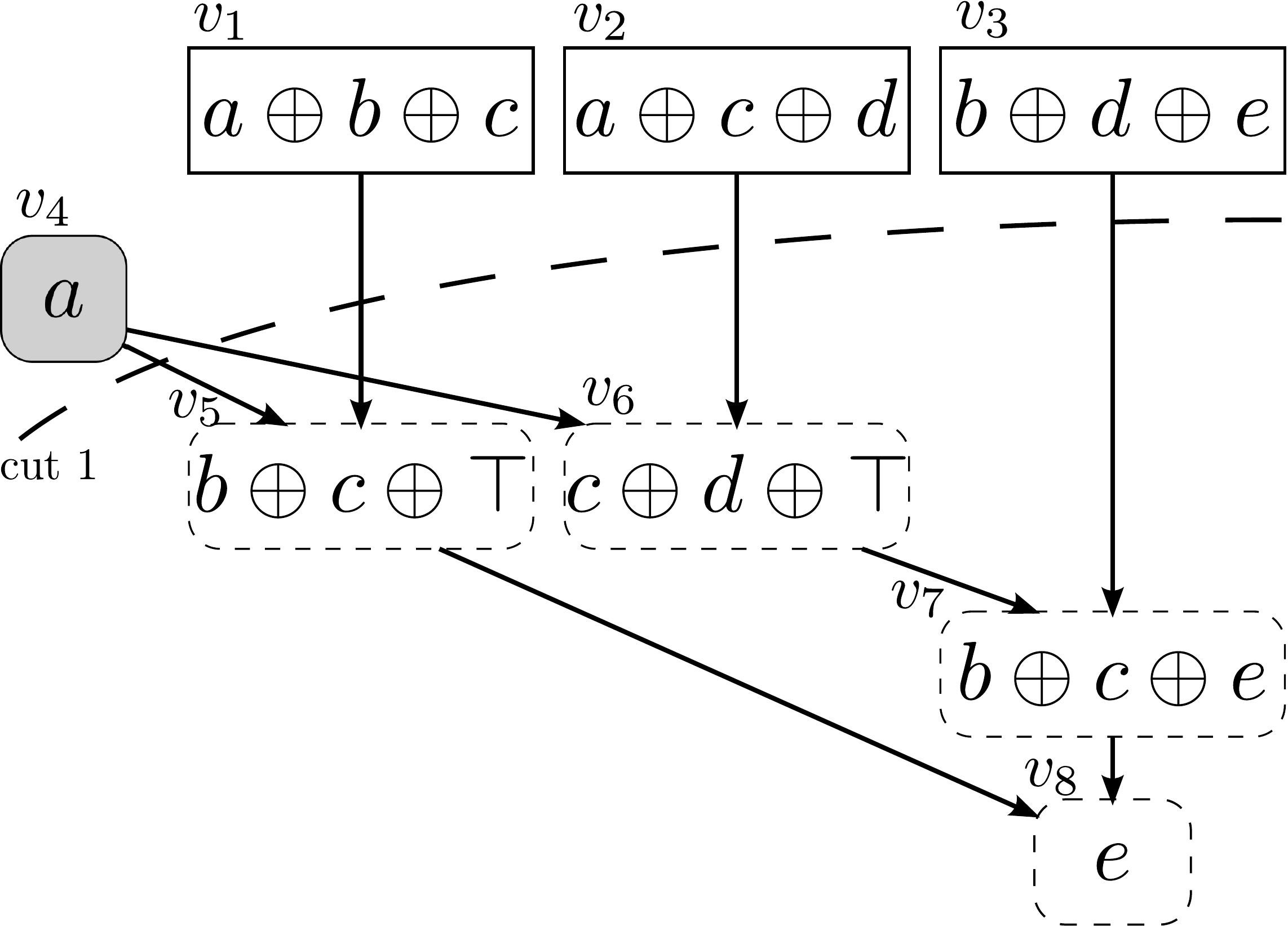}},%
{A $\ECAIsys$-derivation\label{Fig:ECAI}}]
\ECAIsys{}-derivations are defined similarly to \UPsys{}-derivations.
As an example,
Fig.~\ref{Fig:ECAI} shows a $\ECAIsys$-derivation on $ \xorclauses \land (a) $,
where $ \xorclauses = (a \X b \X c) \land (a \X c \X d) \land (b \X d \X e)$.
The literal $e$ is $\ECAIsys$-derivable on
$ \xorclauses \land (a) $; the xor-reasoning module returns
$e$ as an xor-implied literal on $\xorclauses$ after $a$ is given
as an xor-assumption.
The cnf-compatible cut 1 for the literal $e$ gives the implicative explanation
$(a)$ and
thus the implying or-clause $({\neg a} \lor e)$ for $e$.
Parity explanations are defined for $\ECAIsys$ in the same way as for
\UPsys{};
the parity explanation for the literal $e$ in the figure is $\top$
and
thus the implying or-clause for $e$ is $(e)$.
Observe that $e$ is \emph{not} \UPsys{}-derivable from $\xorclauses \land (a)$,
i.e.~\ECAIsys{} is a stronger deduction system than \UPsys{} in this sense.
\end{figwindow}

Parity explanations can also be computed in another xor-reasoning module,
\ECsys{} presented in~\cite{LJN:ICTAI2011},
that is based on equivalence class manipulation.
We omit this construction due to space constraints.

\section{Experimental results}

We have implemented a prototype solver
integrating both xor-reasoning modules (\UPsys{} and \ECAIsys{}) to
\Minisat{} \cite{EenSorensson:2004} (version 2.0 core) solver.
% as described in the section~3.
%
In the experiments we focus on the domain of logical cryptanalysis by modeling
a ``known cipher stream'' attack on stream ciphers Bivium, Crypto-1, Grain,
  Hitag2, and Trivium. To evaluate the performance of the proposed techniques, we include both unsatisfiable and satisfiable instances. 
In the unsatisfiable instances, generated with
\Tool{grain-of-salt}~\cite{SoosTools}, the task is to recover the internal cipher state when 256 output stream bits are given. This is infeasible in practice, so the instances are made easier and
also unsatisfiable by assigning a large enough number of internal state bits
randomly. Thus, the task becomes to prove that it is not possible to assign the
remaining bits of the internal cipher state so that the output would match the
given bits.
To include also satisfiable instances we modeled a different kind of attack on
the ciphers Grain, Hitag2 and Trivium where the task is to recover the full key
when a small number of cipher stream bits are given. In the attack, the IV and
a number of key stream bits are given. There are far fewer generated cipher
stream bits than key bits, so a number of keys probably produce the same prefix
of the cipher stream. 
All instances were converted into (i) the standard DIMACS
CNF format, and (ii) a DIMACS-like format allowing xor-clauses as well.
Structurally these instances are interesting for
benchmarking xor-reasoning as they have a large number of tightly connected
xor-clauses combined with a significant CNF part.

We first compare the following solver configurations:
(i) unmodified \Minisat{},
(ii) \Tool{up}: \Minisat{} with watched variable based unit propagation on xor-clauses,
(iii) \Tool{up-pexp}: \Tool{up} extended with parity explanations,
(iv) \Tool{up-pexp-learn}: \Tool{up-pexp} extended with xor-clause learning,
and
(v) \Tool{up-subst-learn}: \Tool{up} using \ECAIsys{}-module to compute parity xor-explanations and xor-clause learning.
The reference configuration \Tool{up} computes closest cnf-compatible cut
parity explanations, and the other configurations use also furthest cuts
selectively as described in Sect.~\ref{Sec:ParityExplanations}.
We also tested first UIP cuts, but the performance did not improve.

\begin{figure}[t]
  \includegraphics[width=0.49\textwidth,height=27mm]{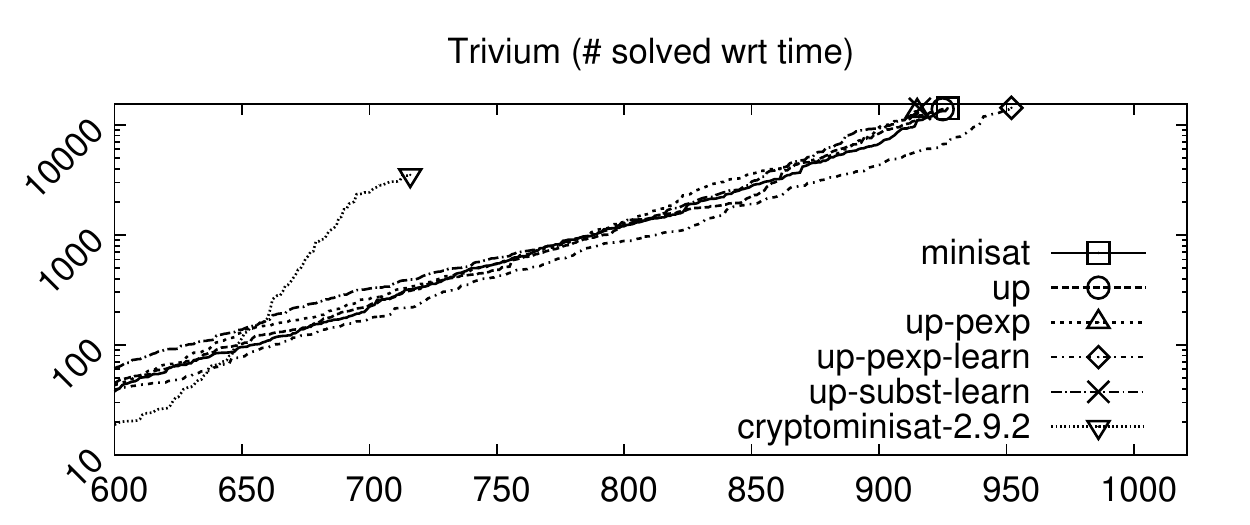}
  \hfill
  \includegraphics[width=0.49\textwidth,height=27mm]{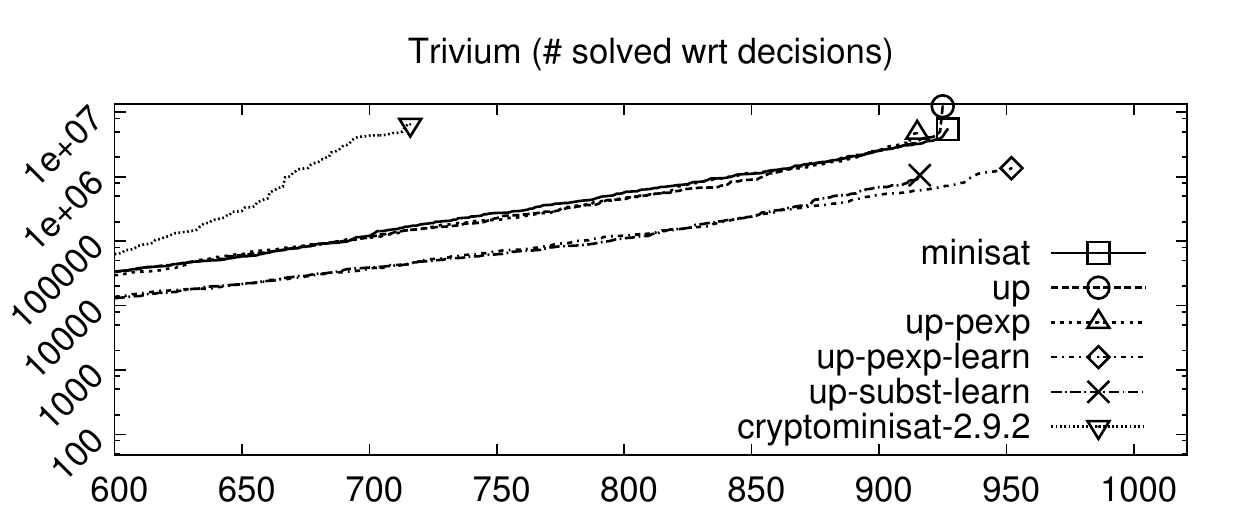}%
  \caption{Number of solved instances with regard to time and decisions
    on satisfiable Trivium benchmark (1020 instances, 51 instances per generated cipher stream length ranging from 1 to 20 bits)}
\label{Fig:Trivium}
\end{figure}

The results for the satisfiable Trivium benchmarks are shown in
Fig.~\ref{Fig:Trivium}.
Learning xor-clauses reduces the number of decisions needed substantially,
and in the case of the computationally less expensive \UPsys{} reasoning module
this is also reflected in the solving time and
in the number of solved instances.
On the other satisfiable benchmark sets learning new xor-clauses also reduced 
the number of required decisions significantly but
the number of propagations per decision is also greatly increased due to increased deduction power
and the reduction is not really reflected in the solving time.

The results for the unsatisfiable benchmarks are shown in
Fig.~\ref{Fig:UnsatResults}. 
Parity explanations reduce decisions on Grain and Hitag2,
leading to fastest solving time.
Learning parity explanations reduces explanations on all benchmarks except
Grain and gives the best solving time on Trivium.
Equivalence reasoning seems to reduce decisions slightly with the cost of increased solving time.
Obviously more work has to be done to improve data structures and
adjust heuristics
so that the theoretical power of parity explanations and xor-clause learning
can be fully seen also in practice.

We also ran \Cryptominisat{} version 2.9.2~\cite{Soos} on the benchmarks.
As shown in Figs.~\ref{Fig:Trivium} and \ref{Fig:UnsatResults},
it performs
(i) extremely well on the unsatisfiable Crypto-1 and Hitag2 instances
due to ``failed literal detection'' and other techniques,
but
(ii) not so well on our satisfiable Trivium instances,
probably due to differences in restart policies or other heuristics.
%
%Since \Cryptominisat{} is based on a newer version of \Minisat{} and
%uses different techniques, heuristics, and optimizations,
%more detailed comparison is omitted here.
%%detailed comparisons with our solver configurations are not very informative. 

\iffalse %ORIGINAL
We also ran \Cryptominisat{} version 2.9.2~\cite{Soos} on the benchmarks.
It performs extremely well except for the satisfiable Trivium instances in which
it performed rather poorly.
Since \Cryptominisat{} is based on a newer version of \Minisat{} and
uses different techniques, heuristics, and optimizations,
more detailed comparison is omitted here.
%detailed comparisons with our solver configurations are not very informative. 
\fi %ORIGINAL

\begin{figure*}[t]
\centering
{\scriptsize
\begin{tabular}{|l|r@{\ \ }r@{\ \ }r|r@{\ \ }r@{\ \ }r|r@{\ \ }r@{\ \ }r|r@{\ \ }r@{\ \ }r|r@{\ \ }r@{\ \ }r|}
\hline &

\multicolumn{3}{c|}{Bivium} &
\multicolumn{3}{c|}{Crypto-1} &
\multicolumn{3}{c|}{Grain} &
\multicolumn{3}{c|}{Hitag2} & 
\multicolumn{3}{c|}{Trivium} \\
Solver &
\# & Dec. & Time &
\# & Dec. & Time &
\# & Dec. & Time &
\# & Dec. & Time &
\# & Dec. & Time \\
\hline
\SMinisat{} & $\bf 51$ & $834$ & $\bf 80.9$ 
& $\bf 51$ & $781$ & $691.1$ 
& $1$ & -  & - 
& $35$ & $428$ & $440.0$ 
& $\bf 51$ & $55$ & $5.7$ 
\\
\STool{up} & $\bf 51$ & $985$ & $127.7$ 
& $\bf 51$ & $1488$ & $1751.8$ 
& $\bf 51$ & $40$ & $13.8$ 
& $39$ & $291$ & $403.9$ 
& $\bf 51$ & $59$ & $8.0$ 
\\
\STool{up-pexp} & $\bf 51$ & $1040$ & $147.8$ 
& $\bf 51$ & $1487$ & $1748.2$ 
& $\bf 51$ & $\bf 35$ & $10.9$ 
& $37$ & $124$ & $148.0$ 
& $\bf 51$ & $62$ & $9.8$ 
\\
\STool{up-pexp-learn} & $47$ & $651$ & $114.0$ 
& $\bf 51$ & $1215$ & $1563.0$ 
& $36$ & $122$ & $87.7$ 
& $37$ & $222$ & $255.4$ 
& $\bf 51$ & $\bf 24$ & $\bf 3.7$ 
\\
\STool{up-subst-learn} & $47$ & $616$ & $336.4$ 
& $\bf 51$ & $1037$ & $2329.5$ 
& $37$ & $70$ & $90.3$ 
& $36$ & $215$ & $374.8$ 
& $\bf 51$ & $29$ & $12.9$ 
\\
\STool{cryptominisat-2.9.2} & $\bf 51$ & $\bf 588$ & $89.8$ 
& $\bf 51$ & $\bf 0$ & $\bf 0.06$ 
& $\bf 51$ & $89$ & $\bf 10.4$ 
& $\bf 51$ & $\bf 0$ & $\bf 0.07$ 
& $\bf 51$ & $71$ & $6.04$ 
\\
\hline
\end{tabular}
}
\caption{Results of the unsatisfiable benchmarks showing the number of solved instances (\#) within the 4h time limit, median of decisions ($\times 10^3$), and median of solving time}
\label{Fig:UnsatResults}
\end{figure*}

\section{Conclusions}

We have shown how to compute linearity exploiting
parity explanations for literals deduced in an xor-reasoning module.
Such explanations can be used
(i) to produce more compact clausal explanations
for the conflict analysis engine of a CDCL solver
incorporating the xor-reasoning module,
and
(ii) to derive new parity constraints that can be learned in order to
boost the deduction power of the xor-reasoning module.
It has been proven that parity explanations allow very short refutations
of some formulas whose CNF translations do not have polynomial size resolution
proofs,
even when using a simple xor-reasoning module capable only of unit-propagation.
The experimental evaluation suggests that parity explanations and xor-clause
learning can be efficiently implemented and demonstrates promising performance
improvements also in practice.

\bibliographystyle{splncs}
\bibliography{paper}

\newpage
\appendix
\section{Proofs}

\renewcommand{\thetheorem}{\ref{Thm:PExpCorrectness}}
\begin{theorem}
  Let $\igraph = \Tuple{ \vertices, \edges, \Labfunc}$ be a \UPsys-derivation
  %for $\Pair{\xorclauses}{\Set{\AL_1,...,\AL_k}}$,
  on $\xorclauses \wedge \AL_1 \wedge \dots \wedge \AL_k$,
  $v$ a node in it,
  and
  $\Cut = \Pair{\VA}{\VB}$ a cut for $v$.
  %It holds that $\xorclauses \Models (\PExpl{v,\Cut} \Iff \Labfunc(v))$.
  %It holds that $\xorclauses \Models (\PExpl{v,\Cut} \Iff \Lab{v})$.
  It holds that $\xorclauses \Models {\PExpl{v,\Cut} \Iff \Lab{v}}$.
\end{theorem}
\begin{proof}
  We show that $\xorclauses \Models (\tmpf{u}{\Cut} \Iff \Lab{u})$ for each
  node $u$ in $\igraph$ by induction on the structure of the derivation
  $\igraph$.

  If $u$ is an input node with $\Lab{u} \in \xorclauses$,
  then $\tmpf{u}{\Cut} = \top$ and
  $\xorclauses \Models (\top \Iff \Lab{u})$
  as $\Lab{u}$ is a conjunct in $\xorclauses$.

  If $u$ is an input node with $\Lab{u} \in \Set{\AL_1,...,\AL_k}$,
  then $\tmpf{u}{\Cut} = \Lab{u}$ and
  $\xorclauses \Models (\Lab{u} \Iff \Lab{u})$ holds trivially.

  If $u$ is a non-input node in $\VA$, then 
  $\tmpf{u}{\Cut} = \Lab{u}$ and $\xorclauses
  \Models (\Lab{u} \Iff \Lab{u}) $ holds trivially.

  If $u$ is a non-input node in $\VB$ with two incoming edges from nodes $u_1$
  and $u_2$ respectively,
  then
  $\tmpf{u}{\Cut} = \tmpf{u_1}{\Cut} \X \tmpf{u_2}{\Cut} \X \top$.
  Because $\Lab{u}$ is obtained from $\Lab{u_1}$ and $\Lab{u_2}$ by applying
  either $\unitruleP$ or $\unitruleN$,
  and the Eqs.~\eqref{Eq:PEP} and \eqref{Eq:PEN} are valid,
  we have $\xorclauses \Models ({\Lab{u_1} \X \Lab{u_2} \X \top} \Iff \Lab{u})$.
  By the induction hypothesis the equations
  $\xorclauses \Models ( \tmpf{u_1}{\Cut} \Iff \Lab{u_1})$
  and $\xorclauses \Models ( \tmpf{u_2}{\Cut} \Iff \Lab{u_2})$ hold,
  so we have
  $\xorclauses \Models (\tmpf{u_1}{\Cut} \X  \tmpf{u_2}{\Cut} \X \top \Iff \Lab{u})$,
  i.e.~$\xorclauses \Models (\tmpf{u}{\Cut} \Iff \Lab{u})$.
  
  For each $u$ in $G$,
  it holds that $\xorclauses \Models (\tmpf{u}{\Cut} \Iff \Lab{u})$.
  It follows that $\xorclauses \Models {\PExpl{v,\Cut} \Iff \Lab{v}}$.
  \qed  
\end{proof}

\renewcommand{\thetheorem}{\ref{Thm:PExpDisjunction}}
\begin{theorem}
  Let $\igraph = \Tuple{ \vertices, \edges, \Labfunc}$ be a \UPsys-derivation
  on $\xorclauses \land {\AL_1 \land \dots \land \AL_k}$,
  $v$ a node with $\Lab{v}=\IL$ in it,
  and
  $\Cut = \Pair{\VA}{\VB}$ a cnf-compatible cut for $v$.
  Then
  $\xorclauses \Models {(\bigwedge_{u \in S} \Lab{u}) \Implies \IL}$,
  where $S = \Setdef{u \in \RS{\Cut}}{\VarsOf{\Lab{u}} \subseteq \VarsOf{\PExpl{v,\Cut}}}$.
\end{theorem}
\begin{proof}
  If $\xorclauses$ is unsatisfiable,
  then $\xorclauses \Models {(\bigwedge_{u \in S}\Lab{u}) \Implies \IL}$
  holds trivially.

  Assume that $\xorclauses$ is satisfiable.
  As $\VarsOf{\bigwedge_{u \in S}\Lab{u}} = \VarsOf{\PExpl{v,\Cut}}$
  and
  $\xorclauses \Models {\PExpl{v,\Cut} \Iff \IL}$ holds,
  it must be that either
  $\xorclauses \Models {(\bigwedge_{u \in S}\Lab{u}) \Implies \IL}$
  or
  $\xorclauses \Models {(\bigwedge_{u \in S}\Lab{u}) \Implies {\neg\IL}}$
  holds
  %(both cannot hold because $\xorclauses$ is satisfiable).
  (both cannot hold because then $\xorclauses$ would be unsatisfiable).
  If $\xorclauses \Models {(\bigwedge_{u \in S}\Lab{u}) \Implies \neg\IL}$
  holds,
  then,
  as $S \subseteq \RS{\Cut}$ holds,
  $\xorclauses \Models {(\bigwedge_{u \in \RS{\Cut}} \Lab{u}) \Implies {\neg\IL}}$
  would also hold.
  Combined with the property
  $\xorclauses \Models {(\bigwedge_{u \in \RS{\Cut}} \Lab{u}) \Implies {\IL}}$
  of implicative explanations,
  this means that $\xorclauses$ is unsatisfiable, contradicting our assumption.
  Therefore,
  $\xorclauses \Models {(\bigwedge_{u \in S}\Lab{u}) \Implies \IL}$
  must hold.
  \qed
\end{proof}
%MUU

%MYY

\renewcommand{\thetheorem}{\ref{Thm:PExpUrquhart}}
\begin{theorem}
  Let $G = \Tuple{V,E}$ be a %3-regular
  \urquhartgraph{} such that $c(G) = \True$.
  There is a \UPsys-refutation for $\urquhartxorclauses(G) \wedge q_1 \dots \wedge q_k $ for some xor-assumptions $ q_1,\dots,q_k$,
  a node $v$ with $\Lab{v}=\bot$ in it,
  and a cut $\Cut = \Tuple{\VA,\VB} $ for $v$
  %such that the parity explanation $\PExpl{v,\Cut}$ is the empty clause.
  such that $\PExpl{v,\Cut} = \top$.
  Thus $\urquhartxorclauses(G) \Models (\top \Iff \bot)$,
  showing $\urquhartxorclauses(G)$ unsatisfiable.
\end{theorem}
\begin{proof}
  Let $E' \subseteq E$ be a spanning tree of $G$,
  and $q_1, \dots, q_k$ the variables occurring on the labels of
  the edges in $E \backslash E'$.
  We construct a \UPsys{}-refutation $G'$ for
  $\urquhartclauses(G) \wedge q_1 \land \dots \wedge q_k$ by
  applying $\unitruleP$ twice for each variable $q_i$.
  Note for every leaf node $u$ in the spanning tree it holds that
  $C = \simplification{\simplification{\urquhartxorclause(u)}{q_1}{\top}\dots}{q_k}{\top}$ is a unit xor-clause $C$ fixing the value of some variable $x$.
  We pick a leaf node $u$,
  propagate value of the variable $x$ %in $C$
  by applying the rule
  $\unitruleP$ or $\unitruleN$,
  and remove the node $u$ from the spanning tree.
  %
  %We pick a leaf node $u$,
  %eliminate both occurrences of the variable $x$ %in $C$
  %by applying the rule
  %$\unitruleP$ or $\unitruleN$,
  %and remove the node $u$ from the spanning tree.
  %
  We proceed until all nodes in the spanning tree are eliminated and we have derived an empty xor-clause $\Labfunc(v) = \bot$ for a node $v$ in the \UPsys{}-refutation $G'$.
  Take the furthest cut $\Cut$.
  %
  %The cut $\Cut$ is acquired by selecting the nodes corresponding to
  %the xor-assumptions  $q_1,\dots,q_k$ for the reason part and
  Computing the parity explanation for $v$ we will get $\PExpl{v,\Cut} = \top$
  as
  there are exactly two paths from  each xor-assumption node $q_i$ to $v$
  and an odd number of non-input nodes.
  %There are exactly two paths from $v'$ to each xor-assumption $q_i$,
  %so the parity explanation $\PExpl{v',\Cut} $ is the empty clause.
  The construction is illustrated in Fig.~\ref{Fig:Urquhart}.
  \qed
\end{proof}

\begin{figure}%[tb]
  \centering
  %\begin{tabular}{c@{\quad\qquad}c}
    \begin{tabular}{@{}c@{}}
    \includegraphics[scale=0.8]{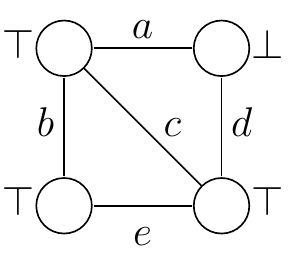} \\
    \small (a)
    \\
    \\
    \includegraphics[scale=0.8]{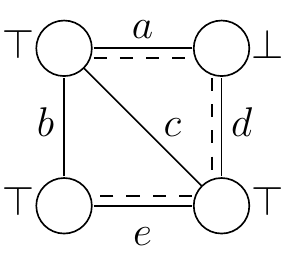} \\
    \small (b) \\ \\
    \end{tabular}
    \hfill
    \begin{tabular}{@{}c@{}}
    \includegraphics[scale=0.8]{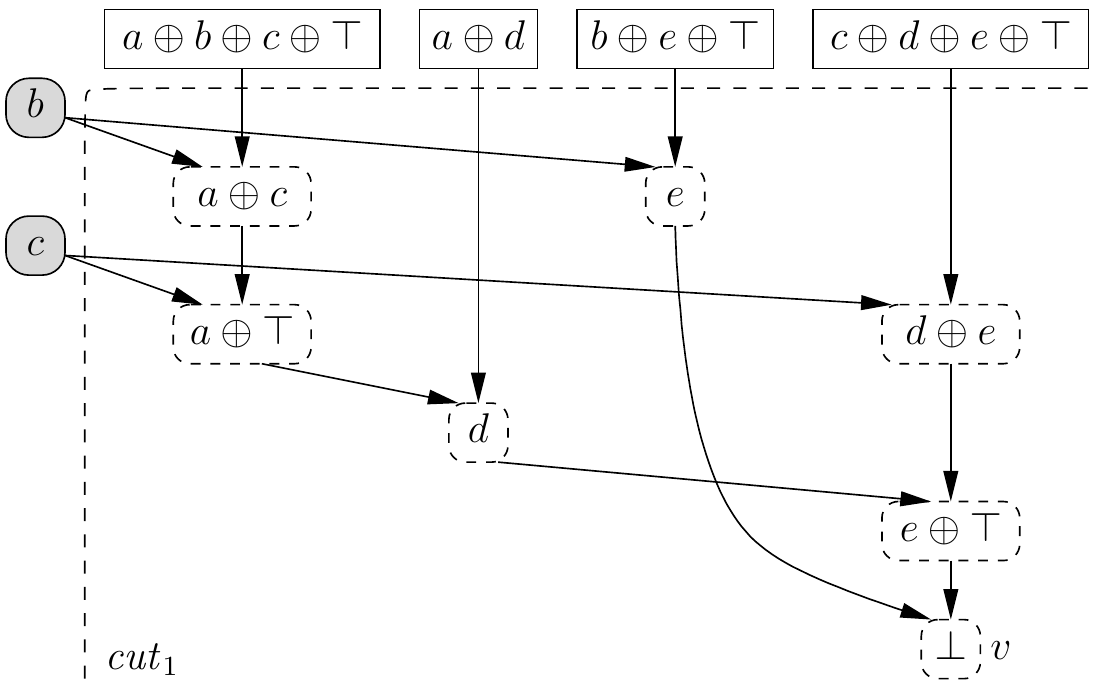} \\
    \small (c)
    \end{tabular}
  %\end{tabular}
  \caption{(a) a \urquhartgraph{} $G$, (b) a spanning tree of the graph,
    and (c) a \UPsys-refutation for $\urquhartxorclauses(G) \land b \land c$
    with $\PExpl{v, \mathit{cut}_1} = \top$.}
  \label{Fig:Urquhart}
\end{figure}

\iffalse
\begin{figure}%[tb]
  \centering
  \begin{tabular}{ccc}
    \includegraphics[scale=0.7]{ex-urquhart} &
    \includegraphics[scale=0.7]{ex-urquhart-spa} &
    \includegraphics[scale=0.7]{ex-urquhart-subst} \\
    \small (a) & \small (b) & \small (c)
  \end{tabular}
  \caption{(a) an \urquhartgraph{} $G$, (b) a spanning tree of the graph,
    and (c) a \UPsys-refutation for $\urquhartxorclauses(G) \land b \land c$
    with $\PExpl{v, \mathit{cut}_1} = \top$.}
  \label{Fig:Urquhart}
\end{figure}
\fi

\end{document}